\documentclass[12pt]{article}

\usepackage{setspace}

 \usepackage[dvips]{graphicx}
 \usepackage{times,amsfonts,amsmath}
 \usepackage{psfrag,ifthen}
 \usepackage{graphics}
 \usepackage{mychicago} 

\usepackage{float}
\makeatletter

\renewcommand\section{\@startsection {section}{1}{\z@}
                                   {-2.5ex \@plus -1ex \@minus -.2ex}
                                   {1.ex \@plus.2ex}
                                   {\normalfont\bfseries\uppercase}}
\renewcommand\subsection{\@startsection{subsection}{2}{\z@}%
                                   {-2.5ex \@plus -1ex \@minus -.2ex}
                                   {1.ex \@plus.2ex}
                                   {\normalfont\bfseries}}

\newcommand{\R}{{I\!\!R}}

\def\R{{\rm I}\! {\rm R}}

\def\X{{\bf X}}

\def\pt{\frac{\partial}{\partial t}}

\renewcommand{\vec}[1]{\mbox{\boldmath $ #1 $}}

\newtheorem{algorithm}{Algorithm}
\newtheorem{remark}{Remark}
\newtheorem{theorem}{Theorem}
\newtheorem{proof}{Proof}

\makeatother

\begin{document}
 \vspace*{12pt} 

\begin{center}
 \textbf{\large{Simulation of a Heat Transfer in Porous Media}}

 \vspace*{24pt}

 \normalsize{J. GEISER$^1$}\\

 \vspace*{12pt}

 $^1$\textit{\normalsize{EMA University of Greifswald, Institute of Physics, Felix-Hausdorff-Str. 6, D-17489 Greifswald, Germany}}\\
 \textit{\normalsize{Email: juergen.geiser@uni-greifswald.de}}\\

  \vspace*{12pt}

\end{center}

 \vspace*{12pt}

 \vspace*{12pt}
{\bf ABSTRACT}

We are motivated to model a heat transfer to 
a multiple layer regime and their optimization 
for heat energy resources. Such a problem can be modeled
by a porous media with different phases (liquid and solid).

The idea arose of a geothermal energy reservoir 
which can be used by cities, e.g. Berlin.

While hot ground areas are covered to most high populated cites,
the energy resources are important and a shift to 
use such resources are enormous.

We design a model of the heat transport via the flow of water through the
heterogeneous layer of the underlying earth sediments.

We discuss a multiple layer model, based on mobile and immobile zones.

Such numerical simulations help to economize on expensive 
physical experiments and obtain control mechanisms for the delicate 
heating process.



\section{Introduction}

We motivate our research on simulating novel energy resources
in geothermic. 

The heat transfer in permeable and non-permeable layers are models
and we simulate the temperatures in the different layers.

Such simulations allow to predict possible energy resources to
geothermal reservoirs.

For such processes, we present a multi phase and multi-species model, 
see \cite{geiser_book_09}.

The solver methods are fast Runge--Kutta solvers,
whereas the mobile terms are convection--diffusion equations and 
are solved with splitting semi-implicit finite volume methods and characteristic methods, \shortcite{gei_06}.

Such a sequential treatment of the partial differential
equations and ordinary differential equations allow of saving computational
time, while expensive implicit Runge--Kutta methods are reduced to 
the partial operators and fast explicit Runge--Kutta methods are 
for the ordinary operators of the multi phase model.

With various source terms we control the required concentration 
at the final temperature area.

\noindent This paper is outlined as follows.

In Section \ref{modell}, we present our mathematical model 
based on the multiphases.
In Section \ref{disc}, we discuss discretization 
and solver methods with respect to their efficiency and accuracy.
The splitting schemes are discussed in Section \ref{oper}.
The numerical experiments are given in Section \ref{num}.
In Section \ref{concl}, 
we briefly summarize our results.

\section{Mathematical modeling}
\label{modell}

In the model we have included the following 
multiple physical processes, related to the deposition process:
\begin{itemize}
\item Flow field of the fluid: Navier--Stokes equation
\item Transport system of the species:  mobile and immobile phases
\end{itemize}

In the following we discuss the three models separately and combine all 
the models into a multiple physical model.
We assume a two-dimensional domain of the apparatus with isotropic flow fields, see \cite{gobb96}.

\subsection{Flow field}

The conservation of momentum is given by (flow field: Navier--Stokes equation)
\begin{eqnarray}
\label{scalar}
&&  \pt {\bf v} + {\bf v} \cdot \nabla {\bf v}  =  - \nabla p , \; \mbox{in} \; \Omega \times [0, t] \\
&& {\bf v}({\bf x}, t) = {\bf v}_0({\bf x}) ,\;  \mbox{on} \; \Omega , \label{tic}\\ 
&& {\bf v}({\bf x}, t) = {\bf v}_1({\bf x},t) , \;  \mbox{on} \; \partial \Omega  \times [0, t] \label{tbc} ,
\end{eqnarray}
where ${\bf v}$ is the velocity field, $p$ the pressure, ${\bf v}_0$ the initial velocity field and the position vector ${\bf x} = (x_1, x_2)^t \in \Omega \subset \R^{2,+}$.
Furthermore, we assume that the flow is divergence free and the pressure 
is pre-defined.

\subsection{Transport systems (multi phase equations)}
 
We model the heat transfer as 
an underlying medium in the earth layers with mobile and immobile phases.
Here heat transport in the fluid with different species contain
of mobile and immobile concentrations. 
For such a heterogeneous media, we applied our expertise in
modeling multiphase transport through a porous medium.

\begin{figure}[H]
\begin{center} 
\includegraphics[width=14.0cm,angle=-0]{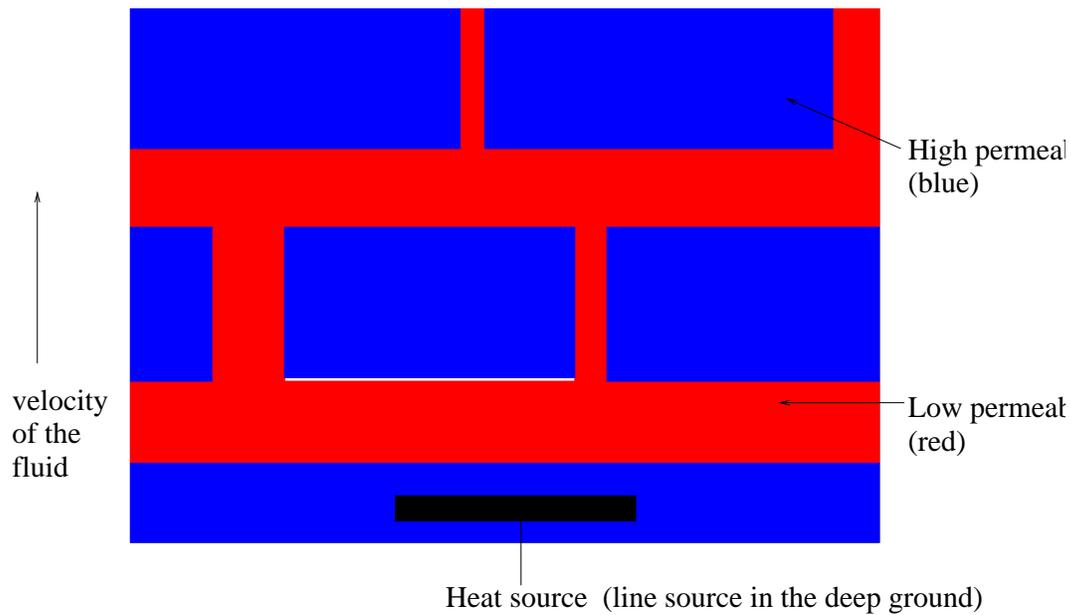}
\end{center}
\caption{\label{part_0} Multiple layer regime of the underlying rocks and earth layers.}
\end{figure}

In the model, we consider both absorption and adsorption 
taking place simultaneously and with given exchange rates.
Therefore we consider the effect of the gas concentrations' being 
incorporated into the porous medium. \\
We extend the model to two more phases:
\begin{itemize}
\item Immobile phase 
\item Adsorbed phase
\end{itemize}

In Figure \ref{part_0_b}, the mobile and immobile phases of the
gas concentration are shown in the macroscopic scale of the porous medium.
Here the exchange rate between the mobile gas concentration and
the immobile gas concentration control the flux to the medium.
\begin{figure}[H]
\begin{center} 
\includegraphics[width=8.0cm,angle=-0]{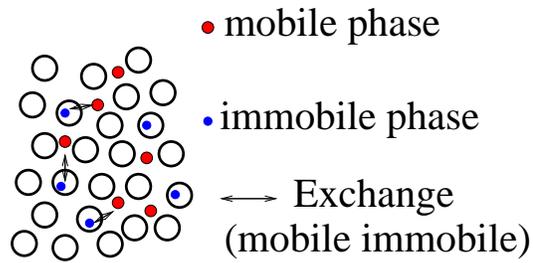}
\end{center}
\caption{\label{part_0_b} Mobile and immobile phase.}
\end{figure}

In Figure \ref{part_0_c}, the mobile and adsorbed phases of the
gas concentration are shown in the macroscopic scale of the porous medium.
To be more detailed in the mobile and immobile phases, where the 
gas concentrations can be adsorbed or absorbed, we consider a further phase. 
Here the adsorption in the mobile and immobile phase is treated as a
retardation and given by a permeability in such layers. 
\begin{figure}[H]
\includegraphics[width=10.0cm,angle=-0]{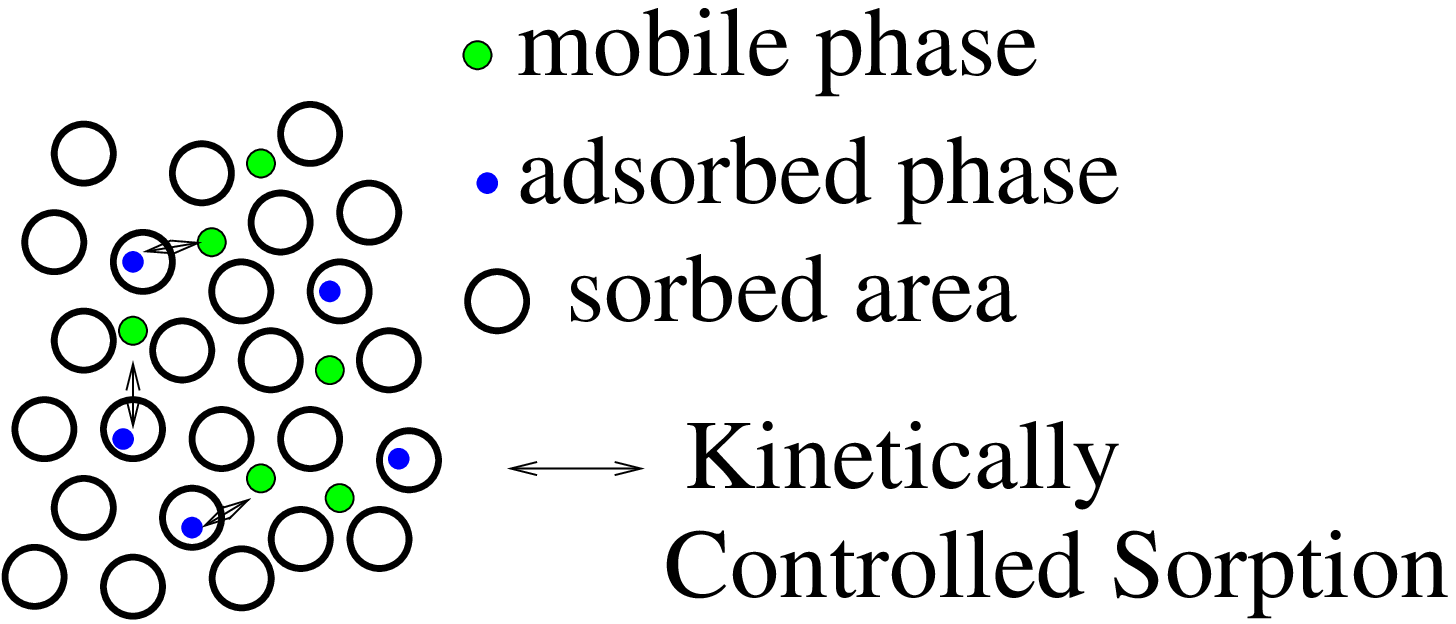}
\hspace{-1.5cm}\includegraphics[width=10.0cm,angle=-0]{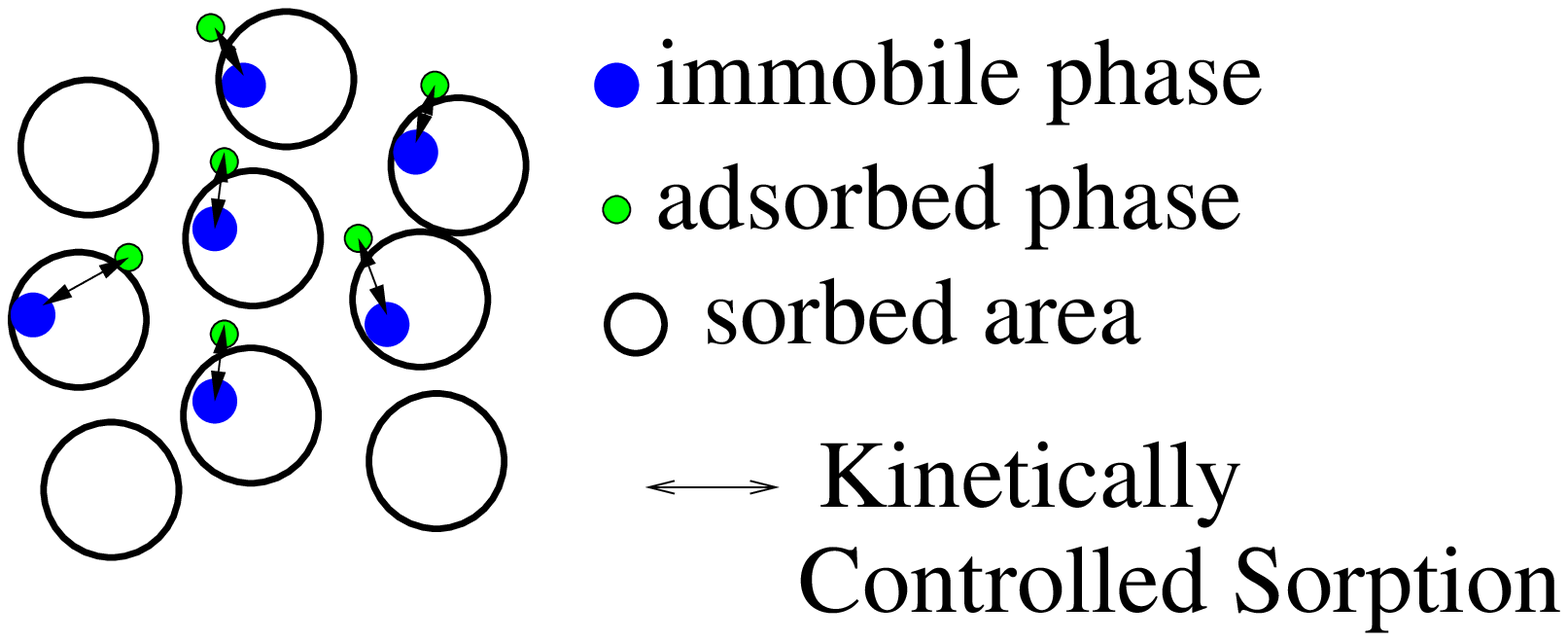}
\caption{\label{part_0_c} Mobile-adsorbed phase and immobile-adsorbed phase.}
\end{figure}

The model equation for the multiple phase equations are 
\begin{eqnarray}
\label{mobile1}
&& \phi \partial_t T_i  + \nabla \cdot {\bf F}_i  =  g (- T_{i}  + T_{i,im} ) + k_{\alpha} (- T_{i}  + T_{i,ad} ) \nonumber \\
&& - \lambda_{i,i} \phi T_i  + \sum_{k=k(i)} \lambda_{i,k} \phi T_k  + \tilde{Q_i} ,  \; \mbox{in} \; \Omega \times [0, t] , \\
&& {\bf F}_i  = {\bf v} T_i  - D^{e(i)} \nabla T_i , \\
\label{immobile1}
&& \phi \partial_t T_{i,im}   =  g (T_{i}  - T_{i,im} ) + k_{\alpha} (T_{i,im,ad}  - T_{i,im} ) \nonumber \\
&& - \lambda_{i,i} \phi T_{i,im}  + \sum_{k=k(i)} \lambda_{i,k} \phi T_{k,im}  
 + \tilde{Q_{i,im}} ,  \; \mbox{in} \; \Omega \times [0, t] , \\
\label{adsorpt1}
&& \phi \partial_t T_{i,ad}   =  k_{\alpha} (T_{i}  - T_{i,ad} ) - \lambda_{i,i} \phi T_{i,ad}  + \sum_{k=k(i)} \lambda_{i,k} \phi T_{k,ad}  
 + \tilde{Q_{i,ad}} ,  \; \mbox{in} \; \Omega \times [0, t]  , \\
\label{adsorpt1_immo}
&& \phi \partial_t T_{i,im,ad}   =  k_{\alpha} (T_{i,im}  - T_{i,im,ad} ) \nonumber \\
&& - \lambda_{i,i} \phi T_{i,im,ad}  + \sum_{k=k(i)} \lambda_{i,k} \phi T_{k,im,ad}  
 + \tilde{Q_{i,im,ad}} , ,  \; \mbox{in} \; \Omega \times [0, t]  ,\\
&& T_i ({\bf x}, t) = c _{i,0}({\bf x}) , T_{i,ad} ({\bf x}, t) = 0 , T_{i,im} ({\bf x}, t) = 0 ,  T_{i,im,ad} ({\bf x}, t) = 0 ,  \;  \mbox{on} \; \Omega , \label{tic}\\ 
&& T_i({\bf x}, t) = T_{i,1}({\bf x},t) , T_{i,ad}({\bf x}, t) = 0 , T_{i,im}({\bf x}, t) = 0 ,  T_{i,im,ad}({\bf x}, t) = 0 , \;  \mbox{on} \; \partial \Omega  \times [0, t] \label{tbc} ,
\end{eqnarray}
where the initial value is given as $T_{i,0}$ and we assume
a Dirichlet boundary conditions with the function $T_{i,1}({\bf x},t)$ sufficiently smooth, all other initial and boundary conditions of the other phases are zero.
\begin{eqnarray*}
\phi :&& \mbox{effective porosity} \; [-], \nonumber \\
 T_i: && \mbox{temperature of the }i\mbox{th species in the underlying rock}
 \nonumber \\
 T_{i,im}: && \mbox{temperature of the }i\mbox{th species in the immobile zones of the rock}
 \nonumber \\
        && \mbox{phase} \; [K / m^3] , \nonumber \\
T_{i,ad}: && \mbox{temperature of the }i\mbox{th species in the adsorbed zones of the rock}
 \nonumber \\
        && \mbox{phase} \; [K / m^3] , \nonumber \\
T_{i,im,ad}: && \mbox{temperature of the }i\mbox{th species in the immobile adsorbed zones of the rock}
 \nonumber \\
        && \mbox{phase} \; [K / m^3] , \nonumber \\
{\bf v }: && \mbox{velocity through the rock and porous substrate \cite{rouch06}}  \; [cm/h]] , \nonumber \\
D^{e(i)} : && \mbox{element-specific diffusion-dispersions tensor} \; [m^2
 / h]]  , \nonumber \\
\lambda_{i,i} : && \mbox{decay constant of the }i\mbox{th species} \; [1 / h]] , \nonumber \\
\tilde{Q}_i : && \mbox{source term of the }i\mbox{th species} \; [K / (m^3 h)]  , \nonumber  \\
g  : && \mbox{exchange rate between the mobile and immobile concentration}  \; [1 / h] , \nonumber \\
k_{\alpha} : && \mbox{exchange rate between the mobile and adsorbed concentration or immobile and} \\
 && \mbox{immobile adsorbed concentration (kinetic controlled sorption)}  \; [1 / h] , \nonumber
\end{eqnarray*}
with $i = 1 , \ldots , M $ and $M$ denotes the number of components.

The parameters in (\ref{mobile1}) are further described, see also \cite{gei_diss03}. \\
The four phases are treated in the full domain, such that we have 
a full coupling in time and space.
 
The effective porosity is denoted by $\phi$ and declares the 
portion of the porosities of the aquifer that is filled with solid grain,
and we assume a nearly solid phase. 
The transport term is indicated by the Darcy velocity ${\bf v}$, 
that presents the flow direction and the absolute value of the heat flux. 
The velocity field is divergence free. The decay constant of the $i$th
species is denoted by $\lambda_i$.
Thereby, $k(i)$ denotes the indices of the other species.

\section{Discretization and solver methods}
\label{disc}

We first discretize the underlying flow and transport equations 
in space with finite volume methods, 
while we then apply the time integration methods,
e.g. Runge-Kutta schemes.

\subsection{Notation}

The time-steps for the calculation in the time-intervals are
$(t^n, t^{n+1}) \subset (0,T)$ , for $n = 0,1, \ldots $. The 
computational cells are given as 
$\Omega_j \subset \Omega$ with $j = 1, \ldots, I$. 
The unknown $I$ is the number of the nodes.

For the application of finite-volumes we have to construct 
a dual mesh for the triangulation {$\cal T$ }, for the domain $\Omega$.
First the finite-elements 
for the domain $\Omega$ are given by $T^e, e = 1 , \ldots , E$.
The polygonal computational cells $\Omega_j$ are related to the 
vertexes $x_j$ of the triangulation. 

The notation for the relation between the neighbor cells and the 
concerned volume of each cell is given in the following notation. \\
Let $V_j = |\Omega_j|$ and the set $\Lambda_j$ denote the
neighbor-point $x_k$ to the point $x_j$. The boundary of the cell $j$ 
and $k$ is denoted as $\Gamma_{jk}$.

We define the flux over the boundary $ \Gamma_{jk}$ as
\begin{eqnarray}\label{eq43_d}
 v_{jk} = \int_{\Gamma_{jk}} {\bf n} \cdot {\bf v} \; ds \; . 
\end{eqnarray}
The inflow-flux is given as $v_{jk} < 0$, and the outflow-flux is 
$v_{jk} > 0$. The antisymmetry of the fluxes is denoted as
$v_{jk} = - v_{kj}$. The total outflow-flux is given as
\begin{eqnarray}\label{eq43_e}
 \nu_j = \sum_{k \in out(j)} v_{jk} .
\end{eqnarray}

The idea of the finite-volumes is to construct an algebraic system
of equation to express the unknowns $c_j^n \approx c(x_j,t^n)$.
The initial values are given by $c_j^0$.
The expression of the interpolation schemes can be given naturally
in two ways:
the first possibility is given with the primary mesh of the 
finite-elements
\begin{eqnarray}\label{eq43_f}
c^n = \sum_{j=1}^I c_j^n \phi_j(x)
\end{eqnarray}
where $\phi_j$ are the standard globally-finite element basis functions
\cite{fro_gei02}.
The second possibility is given with the dual mesh of 
the finite volumes with,
\begin{eqnarray}\label{eq43_g}
\hat{c}^n = \sum_{j=1}^I c_j^n \varphi_j(x)
\end{eqnarray}
where $\varphi_j$ are piecewise constant discontinuous functions
defined by $\varphi_j(x) =1$ for $x \in \Omega_j$ and $\varphi_j(x)=0$
otherwise.

\subsection{Discretization of the Transport equation}

We deal with the transport part, see (\ref{mobile1}):
\begin{eqnarray}
\label{scalar_3}
&& R_i \pt c_i + \nabla \tilde{F_i} = 0 , \; \mbox{in} \; \Omega \times [0, t] \\
&& \tilde{F_i} = {\bf v} c_i - D^{e(i)} \nabla c_i,\nonumber \\
&& c_i(x, t) = c_{i,0}(x) ,\;  \mbox{on} \; \Omega , \label{tic}\\ 
&& c_i(x, t) = c_{i,1}(x,t) , \;  \mbox{on} \; \partial \Omega  \times [0, t] \label{tbc} ,
\end{eqnarray}

For the convection part, we use a piecewise constant 
finite volume method with upwind discretization, see \cite {fro_gei02}.
For the diffusion-dispersion part, we also apply a finite volume
method and we assume the boundary values are denoted by ${\bf n} \cdot D^{e(i)} \; \nabla c_i(x, t) = 0 $, 
where $x \in \Gamma$ is the boundary $\Gamma = \partial \Omega$, 
cf.\ \cite {fro02}. The initial conditions are given by $c_i(x,0) = c_{i,0}(x)$.

We integrate (\ref{scalar_3}) over space and obtain
\begin{eqnarray}
\label{gleich_kap2_20}
\int_{\Omega_j} R_i \pt c_i \; dx = \int_{\Omega_j} \nabla \cdot (- {\bf v} c_i + D^{e(i)} \nabla c_i) \; dx \; .
\end{eqnarray}
The time integration is done later in the decomposition method
with implicit--explicit Runge--Kutta methods. Further the diffusion-dispersion term is lumped, cf.\ \cite {gei_diss03}
Eq.\ (\ref{gleich_kap2_20}) is discretized over
 space using Green's formula.
\begin{eqnarray}
\label{gleich_kap2_21_1}
V_j  R_i \pt c_i \; dx =  \int_{\Gamma_j}
  {\bf n} \cdot (- {\bf v} c_i + D^{e(i)} \nabla c_i) \; d\gamma \; ,
\end{eqnarray}
where $\Gamma_j$ is the boundary of the finite volume 
cell $\Omega_j$ and $V_uj$ is the volume of the cell $j$. 
We use the approximation in space, see \cite {gei_diss03}.

The spatial integration for the diffusion part (\ref{gleich_kap2_21_1}) 
is done by the mid-point rule over its finite boundaries and the
convection part is done with a flux limiter and we obtain:
\begin{eqnarray}
\label{gleich_kap2_23}
V_j R_i \pt c_{i,j}  = \sum_{e \in \Lambda_j} {\bf n}^e  \nabla {\bf v} c_{i}^{e} d \gamma \;  + \;  \sum_{e \in \Lambda_j} \sum_{k \in \Lambda_j^e} |\Gamma_{jk}^e|
{\bf n}_{jk}^e \cdot D_{jk}^e \nabla c_{i,jk}^{e}   \; , 
\end{eqnarray}
where $|\Gamma_{jk}^e|$ is the length of the boundary element $\Gamma_{jk}^e$.
The gradients are calculated with the piecewise finite-element function
$\phi_l$.

We decide to discretize the ﬂux with an up-winding scheme
and obtain the following discretization for the convection part:
\begin{eqnarray}
F_{j, e} = \left\{
\begin{array}{c c}
{\bf v}_{j,e} c_{i,j} & \mbox{if} \; v_{j,e} \ge 0 , \\
{\bf v}_{j,e} c_{i,k} & \mbox{if} \; v_{j,e} < 0 ,
\end{array}
\right.
\end{eqnarray}
where $v_{j,e} = \int_e {\bf v} \cdot {\bf n}_{j,e} ds $.

We obtain for the diffusion part:
\begin{eqnarray}
\label{gleich_kap2_22_2}
\nabla c_{jk}^{e} = \sum_{l \in \Lambda^e} c_l \nabla
\phi_l({\bf x}_{jk}^e) \; .
\end{eqnarray}

We get, using difference notation for the 
neighbor points $j$ and $l$, cf.\ \cite {frodesch01},
the full semi-discretization: 
\begin{eqnarray}
\label{gleich_kap2_24}
&& V_j R_i \pt c_{i, j} =    \sum_{e \in \Lambda_j} \; F_{j, e} + 
 \sum_{e \in \Lambda_j} \;
  \sum_{l \in \Lambda^e \backslash \{j\}} \; \Big( \sum_{k \in \Lambda_j^e}
|\Gamma_{jk}^e| {\bf n}_{jk}^e \cdot D_{jk}^e \nabla \phi_l ({\bf x}_{jk}^e) \Big) (c_j -
c_l)  \; , \nonumber
\end{eqnarray}
where $j = 1, \ldots, m$.

\begin{remark}
For higher order discretization of the convection equation,
we apply a reconstruction which is based on Godunov's method.
We apply a limiter function that fulfills the local min--max property.
The method is explained in \cite {fro_gei02}.
The linear polynomials are reconstructed by the element-wise gradient
and are given by
\begin{eqnarray}
\label{gleich_kap2_16} 
&& u(x_j) = c_j \; , \\
&& \nabla u |_{V_j} = \frac{1}{V_j} \sum_{e = 1}^E  \int_{T^e \cap \Omega_j}
 \nabla c dx \; , \\
&&\mbox{with} \quad j = 1, \ldots, I \; . \nonumber
\end{eqnarray}
The piecewise linear functions are denoted by
\begin{eqnarray}
\label{gleich_kap2_19} 
&& u_{jk} = c_j + \psi_{j} \nabla u |_{V_j} ( x_{jk} - x_j) \; , \\
&& \mbox{with} \quad j = 1, \ldots, I \; , \nonumber 
\end{eqnarray}
where $\psi_j \in (0,1)$ is the limiter function and based on this,
(\ref{gleich_kap2_19}) fulfills the 
discrete minimum maximum property, as described in \cite {fro_gei02}.

\end{remark}

\subsection{Discretization of the source-terms}

The source terms are part of the convection-diffusion equations 
and are given as follows:
\begin{eqnarray}
\label{eq20_source}
&& \partial_t c_i(x,t) - {\bf v} \cdot \nabla c_i + \nabla D \nabla c_i  = q_i(x, t)
\; ,
\end{eqnarray}
where $i = 1, \ldots, m$, ${\bf v}$ is the velocity, $D$ is the diffusion
tensor and $q_i(x,t)$
are the source functions, which can be
point wise, linear in the domain.

The point wise sources are given as :
\begin{eqnarray}
\label{eq20_source_2}
&& q_i(t) = \left\{ 
\begin{array}{c c}
\frac{q_{s,i}}{T}  & t \le T , \\
0 & t > T ,
\end{array} \right. ,
\mbox{with} \int_T q_i(t) dt = q_{s,i} ,
\end{eqnarray}
where $q_{s,i}$ is the concentration of species $i$
at source point $x_{source, i} \in \Omega$ over the whole time-interval.

The line and area sources are given as :
\begin{eqnarray}
\label{eq20_source_2}
&& q_i(x,t) = \left\{ 
\begin{array}{c c}
\frac{q_{s,i}}{T |\Omega_{source, i}|},   & t \le T \; \mbox{and} \; x \in \Omega_{source, i}, \\
0, & t > T ,
\end{array} \right. , \\
&& \mbox{with} \int_{\Omega_{source, i}} \int_T q_i(x, t) dt dx = q_{s,i},  \nonumber
\end{eqnarray}
where $q_{s,i}$ is the source concentration of species $i$
at the line or area of the source over the whole time-interval.

For the finite-volume discretization we have to compute :
\begin{eqnarray}
\label{gleich_source_1}
\int_{\Omega_{source,i,j}} q_i(x,t) \; dx =  \int_{\Gamma_{source,i,j}} {\bf n} \cdot ( {\bf v} c_i - D \nabla c_i )\; d\gamma \; ,
\end{eqnarray}
where $\Gamma_{source,i,j}$ is the boundary of the finite-volume 
cell $\Omega_{source,i,j}$ which is a source area.
We have $\cup_{j}\Omega_{source,i,j} = \Omega_{source,i}$ where $j \in I_{source}$, where $I_{source}$ is the set of the finite-volume cells that includes
the area of the source. 

The right-hand side of (\ref{gleich_source_1})  is also called the flux of the sources \cite{fro02_b}.

\subsection{Discretization of the Navier-Stokes equation}

We deal with the following Navier-Stokes equation:
\begin{eqnarray}
\label{scalar}
&&  \pt {\bf v} + {\bf v} \cdot \nabla {\bf v}  =  - \nabla p , \; \mbox{in} \; \Omega \times [0, t] \\
&& \nabla \cdot {\bf v} = 0 ,
\end{eqnarray}
where ${\bf v} = (v_1, v_2)^t$, for simplicity we have normalized with $\rho = 1$,
and $p$ is the pressure which is predefined.

For the time discretization, we use the explicit Euler method given by:
\begin{eqnarray}
\label{scalar}
&&  {\bf v}^{n+1} =  {\bf v}^{n} - \Delta t  {\bf v}^n \cdot \nabla {\bf v}^n  - \Delta t \nabla p^n , \; \mbox{in} \; \Omega \\
&& \nabla \cdot {\bf v}^n = 0 ,
\end{eqnarray}
where $\Delta t$ is the local time step.

For the spatial discretization, we apply finite volume methods on staggered grids
and discretize in each direction of the 2D Cartesian grid.
The convection term in the $v_1$-momentum equation is given by,  see \cite{piller04}:
\begin{eqnarray}  
\label{gleich_kap2_10}
\int_{V_h} v_1 \nabla \cdot {\bf v} \; dV = \int_{S_h} v_1  {\bf v} {\bf n} \; dS ,
\end{eqnarray}
where $V_h$ is the control volume with grid size $h$ and $S_h$ is the underlying boundary.
We integrate over each face of the finite volume respecting the direction 
of the normal vector, see \cite{piller04} and next subsection.

The same procedure is also used for the convection term in the $v_2$ momentum equation.

\subsection{Time discretization methods}

We deal with higher order time-discretization methods.
We apply the Runge-Kutta methods as time-discretization methods 
to reach higher order results.

Based on the spatial discretized transport or flow
equations we obtain the following equations:
\begin{equation}
\begin{array}{c}
{\displaystyle \partial_t c(t) = A c(t) + B c(t) + f(t), \quad 0 < t
\leq T } \; ,  \\
\noalign{\vskip 1ex} {\displaystyle c(0)=c_0 \; , }
\end{array} \label{eq:ACP}
\end{equation}
where $A$ is the stiffness operator and $B$ is the reaction operator 
for the transport equations. $f(t)$ is the right hand side, e.g.
source term of the equations.

For such a system of ordinary differential equations, we apply the
Runge-Kutta methods.

{\bf Runge-Kutta method}

We use the implicit trapezoidal rule:
\begin{eqnarray}
\label{num_8}
\begin{array}{c | c c }
0 &  &    \\[0.05cm]
 1  & \frac{1}{2} & \frac{1}{2} \\[0.05cm]
\hline
    & \frac{1}{2} & \frac{1}{2}
\end{array}
\end{eqnarray}

\begin{remark}
We apply also higher order Runge-Kutta schemes.
Based on the spatial discretisation method, which is second order
finite volume schemes, we obtain the
best results with second order RK schemes.
\end{remark}

\section{Splitting methods}
\label{oper}

In the following, we discuss splitting methods to decouple 
the system of differential equations to simpler parts and
accelerate the solver process. \\
We concentrate on two ideas:
\begin{itemize}
\item Additive Splitting schemes ,
\item Iterative Splitting schemes .
\end{itemize}

\subsection{Additive Splitting schemes}

We deal with the following equation:

\begin{eqnarray}
\label{equ1}
&& \sum_{\beta = 1}^p B_{\alpha \beta} \partial_t \; u_{\beta} = \sum_{\beta= 1}^p A_{\alpha \beta} u_{\beta} + f_{\alpha}, \; \alpha= 1, 2, \ldots, p , \\
&& \; u_{\alpha}(0) = u_{\alpha,0} , \; \alpha= 1, 2, \ldots, p .
\end{eqnarray}

Further we assume $A$ and $B$ are self-adjoint.

We apply the discretization with the schemes of weights and obtain:
\begin{eqnarray}
\label{equ1_splitting}
&&  B  \frac{u^{n+1}-u^n}{\tau} - A (\sigma u^{n+1} + (1 - \sigma) u^n) = \phi^n , \\
&& \phi^n = f (\sigma t^{n+1} + (1 - \sigma) t^n) ,
\end{eqnarray}

By the transition to a new time level, we require:
\begin{eqnarray}
\label{equ1}
&&  (B - A \sigma \tau ) u^{n+1} = \phi^n ,
\end{eqnarray}
the original problem can be transferred to
\begin{eqnarray}
\label{equ1}
&& \sum_{\beta=1}^p (B_{\alpha \beta} - A_{\alpha \beta} \sigma \tau ) u_{\beta}^{n+1} = \phi^n_{\alpha} , \; \alpha = 1, 2, \ldots, p .
\end{eqnarray}

By the conduction to a sequence of simpler problems we 
\begin{eqnarray}
\label{equ1}
&& (B_{\alpha \alpha} - \frac{1}{2} A_{\alpha \alpha} \sigma \tau ) u_{\beta}^{n+1/2} = \tilde{\psi}^n_{\alpha} , \; \alpha = 1, 2, \ldots, p , \\
&& (B_{\alpha \alpha} - \frac{1}{2} A_{\alpha \alpha} \sigma \tau ) u_{\beta}^{n+1} = \hat{\psi}^n_{\alpha} , \; \alpha = 1, 2, \ldots, p , \\
\end{eqnarray}

Here we have the benefit to invert only the diagonal parts of the 
matrices and use the idea to solve the triangular splitting of the operator $A = A_1 + A_2$.

\smallskip
\begin{theorem}

If we choose $\sigma \ge \frac{1}{2}$, then the splitting scheme (\ref{equ1_splitting})
is absolute stable in an appropriate Hilbert space.
\end{theorem}

\begin{proof}
The outline of the proof is given in \cite{vab2011}.
\end{proof}

\subsection{Iterative splitting method}

The following algorithm is based on the iteration with
fixed-splitting discretization step-size $\tau$, namely, on the
time-interval $[t^n,t^{n+1}]$ we solve the following sub-problems
consecutively for $i=0,2, \dots 2m$. (cf. \cite{glow03,kan03}.):

\begin{eqnarray}
 && \frac{\partial c_i(t)}{\partial t} = A_1
c_i(t) \; + \; A_2 c_{i-1}(t), \;
\mbox{with} \; \; c_i(t^n) = c^{n} \label{kap3_iter_1} \\
&& \mbox{and} \; c_{0}(t^n) = c^n \; , \; c_{-1} = 0.0 , \nonumber
\\\label{kap3_iter_2}
&& \frac{\partial c_{i+1}(t)}{\partial t} = A_1 c_i(t) \; + \; A_2 c_{i+1}(t), \; \\
&& \mbox{with} \; \; c_{i+1}(t^n) = c^{n}\; , \nonumber
\end{eqnarray}
 where $c^n$ is the known split
approximation at the time-level $t=t^{n}$. The split
approximation at the time-level $t=t^{n+1}$ is defined as
$c^{n+1}=c_{2m+1}(t^{n+1})$. (Clearly, the function $c_{i+1}(t)$
depends on the interval $[t^n,t^{n+1}]$, too, but, for the sake of
simplicity, in our notation we omit the dependence on $n$.)

\smallskip
In the following we will analyze the convergence and the rate of
convergence of the method
(\ref{kap3_iter_1})--(\ref{kap3_iter_2}) for $m$ tends to
infinity for the linear operators
$A_1, A_2: \! {\X} \rightarrow {\X}$,
where we assume that these operators and their sum are
generators of the $C_0$ semi-groups. We emphasize that these
operators are not necessarily bounded, so the convergence is
examined in a general Banach space setting.

The novelty of the convergence results are the reformulation 
in integral-notation. Based on this, we can assume to have bounded
integral operators which can be estimated and given 
in a recursive form. Such formulations are known in the work of
\cite{han08} and  \cite{jan00} and estimations of the kernel part with the
exponential operators are sufficient to estimate the recursive
formulations.

\subsection{Splitting Method to couple mobile and immobile and adsorbed parts}

The motivation of the splitting method are based on the following observations:
\begin{itemize}
\item The mobile phase is semidiscretised with fast finite volume methods and
can be stored into a stiffness-matrix. We achieve large time steps, if we 
consider implicit Runge-Kutta methods of lower order (e.g. implicit Euler) as
a time discretization method. 
\item The immobile, adsorbed and immobile-adsorbed phases are purely
ordinary differential equations and the each cheap to solve with explicit Runge-Kutta schemes.
\item The ODEs can be seen as perturbations and can be solved all explicit 
in a fast iterative scheme.
\end{itemize}

For the full equation we consider the following matrix notation:
\begin{eqnarray}
\label{eq20}
&& \partial_t {\bf c} =  A_1 {\bf c} + A_2 {\bf c} + B_1 ({\bf c} - {\bf c}_{im}) + B_2  ({\bf c} - {\bf c}_{ad}) + {\bf Q} \; , \\
&& \partial_t {\bf c}_{im} =  A_2 {\bf c}_{im} + B_1 ({\bf c}_{im} - {\bf c}) + B_2  ({\bf c}_{im} - {\bf c}_{im,ad}) + {\bf Q}_{im} \; , \\
&& \partial_t {\bf c}_{ad} =  A_2 {\bf c}_{ad} + B_2  ({\bf c}_{ad} - {\bf c}) + {\bf Q}_{ad} \; , \\
&& \partial_t {\bf c}_{im,ad} = A_2 {\bf c}_{im,ad} + B_2  ({\bf c}_{im,ad} - {\bf c}_{im}) + {\bf Q}_{im,ad} \; ,
\end{eqnarray}
where ${\bf c} = (c_1, \ldots, c_m)^T$ is the spatial discretised concentration 
in the mobile phase, see equation (\ref{mobile1}), ${\bf c}_{im} = (c_{1,im}, \ldots, c_{m,im})^T$ is the concentration in the immobile phase, the some also for the other phase concentrations.
$A_1$ is the stiffness matrix of equation (\ref{mobile1}), $A_2$ is the reaction matrix of the right hand side of (\ref{mobile1}), 
$B_1$ and $B_2$ are diagonal matrices with the
exchange of the immobile and kinetic parameters, see equation (\ref{adsorpt1}) and (\ref{adsorpt1_immo}). 

Further  ${\bf Q}, \ldots, {\bf Q}_{im,ad}$ are the spatial discretised sources vectors.

Now we have the following ordinary differential equation:
\begin{eqnarray}
\label{eq20}
&& \partial_t {\bf C} = 
\left(
\begin{array}{c c c c}
 A_1 + A_2 + B_1 + B_2 & -B_1 & -B_2 & 0   \\
 - B_1  & A_2 + B_1 + B_2 & 0 & -B_2     \\
 - B_2  &  0   & A_2 + B_2 & 0     \\
   0 & - B_2 & 0 & A_2 + B_2 
\end{array}
\right)  {\bf C} + \tilde{{\bf Q}} , 
\end{eqnarray}
where $ {\bf C} = ({\bf c}, {\bf c}_{im}, {\bf c}_{ad}, {\bf c}_{im,ad})^T$ 
and the right hand side is given as $\tilde{{\bf Q}} = ({\bf Q}, {\bf Q}_{im}, {\bf Q}_{ad}, {\bf Q}_{im,ad} )^T $.

For such an equation we apply the decomposition of the matrices:
\begin{eqnarray}
\label{eq20}
&& \partial_t {\bf C} = \tilde{A} {\bf C} + \tilde{{\bf Q}} , \\
&& \partial_t {\bf C} = \tilde{A_1} {\bf C}  + \tilde{A_2} {\bf C} + \tilde{{\bf Q}} ,
\end{eqnarray}
where
\begin{eqnarray}
\label{eq20}
\tilde{A_1} =
\left(
\begin{array}{c c c c}
 A_1 + A_2  & 0 & 0 & 0   \\
 0  & A_2 & 0 & 0     \\
 0  &  0   & A_2  & 0     \\
   0 & 0 & 0 & A_2 
\end{array}
\right),
\label{eq20}
\tilde{A_2} =
\left(
\begin{array}{c c c c}
 B_1 + B_2 & -B_1 & -B_2 & 0   \\
 - B_1  & B_1 + B_2 & 0 & -B_2     \\
 - B_2  &  0   &  B_2 & 0     \\
   0 & - B_2 & 0 &  B_2 
\end{array}
\right),
\end{eqnarray}

The equation system is numerically solved by an iterative scheme:
\begin{algorithm}
\label{algo1}
We divide our time interval $[0, T]$ into sub-intervals $[t^n, t^{n+1}]$, where $n=0,1,\dots N$, $t^0=0$ and $t^N=T$.

We start with $n = 0$:

1.) The initial conditions are given with ${\bf C}_{0}(t^{n+1}) = {\bf C}(t^{n})$. We start with $k = 0$.

2.) Compute the fix point iteration scheme given as:
\begin{eqnarray}
\label{eq20_iter}
&& \partial_t {\bf C}^k =  \tilde{A}_1 {\bf C}^k + \tilde{A}_2 {\bf C}^{k-1} + \tilde{{\bf Q}} \; ,
\end{eqnarray}
where $k$ is the iteration index, see \cite{fargei05}.
For the time integration, we apply Runge-Kutta methods as ODE solvers, see \cite{hai92} and \cite{hai96}.

3.) The stop criterion for the time interval $[t^n, t^{n+1}]$ is given as:
\begin{eqnarray}
\label{eq_stop}
&& || {\bf C}^k(t^{n+1}) -  {\bf C}^{k-1}(t^{n+1})|| \le err , 
\end{eqnarray}
where $||\cdot||$ is the maximum norm over all components of the solution
vector. $err$  is a given error bound, e.g. $err = 10^{-4}$.

If equation (\ref{eq_stop}) is fulfilled, we have the result 
\begin{eqnarray}
\label{eq_res}
{\bf C}(t^{n+1}) = {\bf C}^k(t^{n+1}),
\end{eqnarray} 

If $n = N$ then we stop and are done.

If equation (\ref{eq_stop}) is not fulfilled, we do $k = k+1$ and go-to 2.).

\end{algorithm}

The error analysis of the schemes are given in the following Theorem:
\smallskip
\begin{theorem}{\label{Th1}}
Let $A,B \in {\mathcal L(\X)} $ be given linear bounded operators in a Banach space ${\mathcal L(\X)}$. We consider the abstract Cauchy problem:
\begin{eqnarray}
\label{kap3_iter_3}
 \partial_t {\bf C}(t) & = & \tilde{A} {\bf C}(t) + \tilde{B} {\bf C}(t), \quad t_n \leq t \leq t_{n+1} ,  \\
 {\bf C}(t_n) & = & {\bf C}_n,  \; \mbox{for} \; n = 1, \ldots, N ,
\end{eqnarray}
where $t_1 = 0$ and the final time is $t_N = T \in \R^+$.
\noindent
 Then problem {\rm(\ref{kap3_iter_3})} has a unique solution.  For a finite steps with time size $\tau_n = t^{n+1} - t^n$, the iteration
{\rm(\ref{eq20_iter})} for \\
$k=1,2, \dots , q$ is consistent with an order  of consistency ${\mathcal
O} (\tau_n^{q})$.
\end{theorem}

\begin{proof}
The outline of the proof is given in \cite{geiser_book_09}.
\end{proof}

\section{Numerical Experiments}
\label{num}

In the following, we present to heat-flow problems.

\subsection{Two phase example}

The next example is a simplified real-life problem 
for a multiphase transport-reaction equation.
We deal with mobile and immobile pores in the porous media,
such simulations are given for heat transfers in earth layers.

We concentrate on the computational benefits of a fast
computation of the iterative scheme, given with matrix exponential. \\

The equation is given as:
\begin{eqnarray}
\label{mobile1_1_2}
&& \partial_t c_1  + \nabla \cdot {\bf F} c_1  =  g (- c_{1}  + c_{1,im} ) - \lambda_1 c_1 ,  \; \mbox{in} \; \Omega \times [0, t] , \\
&& \partial_t c_2  + \nabla \cdot {\bf F} c_2  =  g (- c_{2}  + c_{2,im} )  + \lambda_1 c_1 - \lambda_2 c_2  ,  \; \mbox{in} \; \Omega \times [0, t] , \\
&& {\bf F}  = {\bf v}  - D \nabla , \\
\label{immobile1_1_2}
&& \partial_t c_{1,im}   =  g (c_{1}  - c_{1,im} ) - \lambda_1  c_{1,im} ,  \; \mbox{in} \; \Omega \times [0, t] , \\
&& \partial_t c_{2,im}   =  g (c_{2}  - c_{2,im} ) + \lambda_1  c_{1,im} - \lambda_2  c_{2,im} ,  \; \mbox{in} \; \Omega \times [0, t] , \\
&& c_1({\bf x}, t) = c _{1,0}({\bf x}),  c_2({\bf x}, t) = c _{2,0}({\bf x}) , \; \mbox{on} \; \Omega , \\
&& c_1({\bf x}, t) = c_{1,1}({\bf x},t) , c_2({\bf x}, t) = c_{2,1}({\bf x},t) , \;  \mbox{on} \; \partial \Omega  \times [0, t] , \\
&& c_{1,im} ({\bf x}, t) = 0 ,  c_{2,im} ({\bf x}, t) = 0 ,  \;  \mbox{on} \; \Omega , \\ 
&&  c_{1,im}({\bf x}, t) = 0 ,  c_{2,im}({\bf x}, t) = 0 , \;  \mbox{on} \; \partial \Omega  \times [0, t] ,
\end{eqnarray}

In the following we deal with the semidiscretized equation given with the 
matrices:
\begin{eqnarray}
\label{eq20}
&& \partial_t {\bf C} = 
\left(
\begin{array}{c c c c}
 A - \Lambda_1 -G & 0 & G & 0   \\
 \Lambda_1  & A - \Lambda_2 - G & 0 & G     \\
 G  &  0   & -\Lambda_1 - G & 0     \\
   0 & G & \Lambda_1 & -\Lambda_2 - G 
\end{array}
\right)  {\bf C} , 
\end{eqnarray}
where $ {\bf C} = ({\bf c_1}, {\bf c_2}, {\bf c_1}_{im}, {\bf c_2}_{im})^T$,
while ${\bf c_1} = (c_{1,1}, \ldots, c_{1,I})$ is the solution of
the first heat species in the mobile phase in each spatial discretization point (i = 1, \ldots, I), the same is also for the other solution vectors.

We have the following two operators for the splitting method:
\begin{eqnarray}
A & = &  \frac{D}{\Delta x^2}\cdot  \left(\begin{array}{rrrrr}
 -2 & 1 & ~ & ~ & ~ \\
 1 & -2 & 1 & ~ & ~ \\
 ~ & \ddots & \ddots & \ddots & ~ \\
 ~ & ~ & 1 & -2 & 1 \\
 ~ & ~ & ~ & 1 & -2
\end{array}\right) \\[6pt]
 & + & \frac{v}{\Delta x}\cdot \left(\begin{array}{rrrrr}
 1 & ~ & ~ & ~ & ~ \\
 -1 & 1 & ~ & ~ & ~ \\
 ~ & \ddots & \ddots & ~ & ~ \\
 ~ & ~ & -1 & 1 & ~ \\
 ~ & ~ & ~ & -1 & 1
\end{array}\right)~\in~\R^{I \times I}
\end{eqnarray}
where $I$ is the number of spatial points.
\begin{eqnarray}
\Lambda_1 & = &    \left(\begin{array}{rrrrr}
 \lambda_1 & 0 & ~ & ~ & ~ \\
        0 &  \lambda_1 & 0 & ~ & ~ \\
 ~ & \ddots & \ddots & \ddots & ~ \\
 ~ & ~ & 0 & \lambda_1 & 0 \\
 ~ & ~ & ~ & 0 & \lambda_1
\end{array}\right) ~\in~\R^{I \times I}
\end{eqnarray}
\begin{eqnarray}
\Lambda_2 & = &    \left(\begin{array}{rrrrr}
 \lambda_2 & 0 & ~ & ~ & ~ \\
        0 &  \lambda_2 & 0 & ~ & ~ \\
 ~ & \ddots & \ddots & \ddots & ~ \\
 ~ & ~ & 0 & \lambda_2 & 0 \\
 ~ & ~ & ~ & 0 & \lambda_2
\end{array}\right) ~\in~\R^{I \times I}
\end{eqnarray}
\begin{eqnarray}
G & = &    \left(\begin{array}{rrrrr}
 g & 0 & ~ & ~ & ~ \\
        0 &  g & 0 & ~ & ~ \\
 ~ & \ddots & \ddots & \ddots & ~ \\
 ~ & ~ & 0 & g & 0 \\
 ~ & ~ & ~ & 0 & g
\end{array}\right) ~\in~\R^{I \times I}
\end{eqnarray}

We decouple into the following matrices:
\begin{eqnarray}
A_1 & = &  \left(\begin{array}{c c c c }
 A & 0 & 0 & 0\\
 0 & A & 0 & 0 \\
 0 & 0 & 0 & 0 \\
 0 & 0 & 0 & 0 
\end{array}\right) ~\in~\R^{4 I \times 4 I}
\end{eqnarray}
\begin{eqnarray}
\tilde{A}_2 & = &  \left(\begin{array}{c c c c}
 - \Lambda_1 & 0 & 0 & 0\\
 \Lambda_1 & - \Lambda_2 & 0 & 0 \\
 0 & 0 & - \Lambda_1 & 0 \\
 0 & 0 & \Lambda_1 & - \Lambda_2 
\end{array}\right) ~\in~\R^{4 I \times 4 I}
\end{eqnarray}
\begin{eqnarray}
\tilde{A}_3 & = &  \left(\begin{array}{c c c c}
 - G & 0 & G & 0\\
  0 & - G & 0 & G \\
 G & 0 & - G & 0 \\
 0 & G & 0 & - G 
\end{array}\right) ~\in~\R^{4 I \times 4 I}
\end{eqnarray}
For the operator $A_1$ and $A_2 = \tilde{A}_2 + \tilde{A}_3$ we apply the
iterative splitting method.

Based on the decomposition, operator $A_1$ is only tridiagonal
and operator $A_2$ is block diagonal. Such matrix structure 
reduce the computation of the exponential operators.

The Figure \ref{two_phase} present the numerical errors between the exact and the
numerical solution. Here we obtain optimal results for one-side iterative schemes on operator $B$, means we iterate with respect to $B$ and use $A$ as right hand side. 
\begin{figure}[ht]
\begin{center}  
\includegraphics[width=9.0cm,angle=-0]{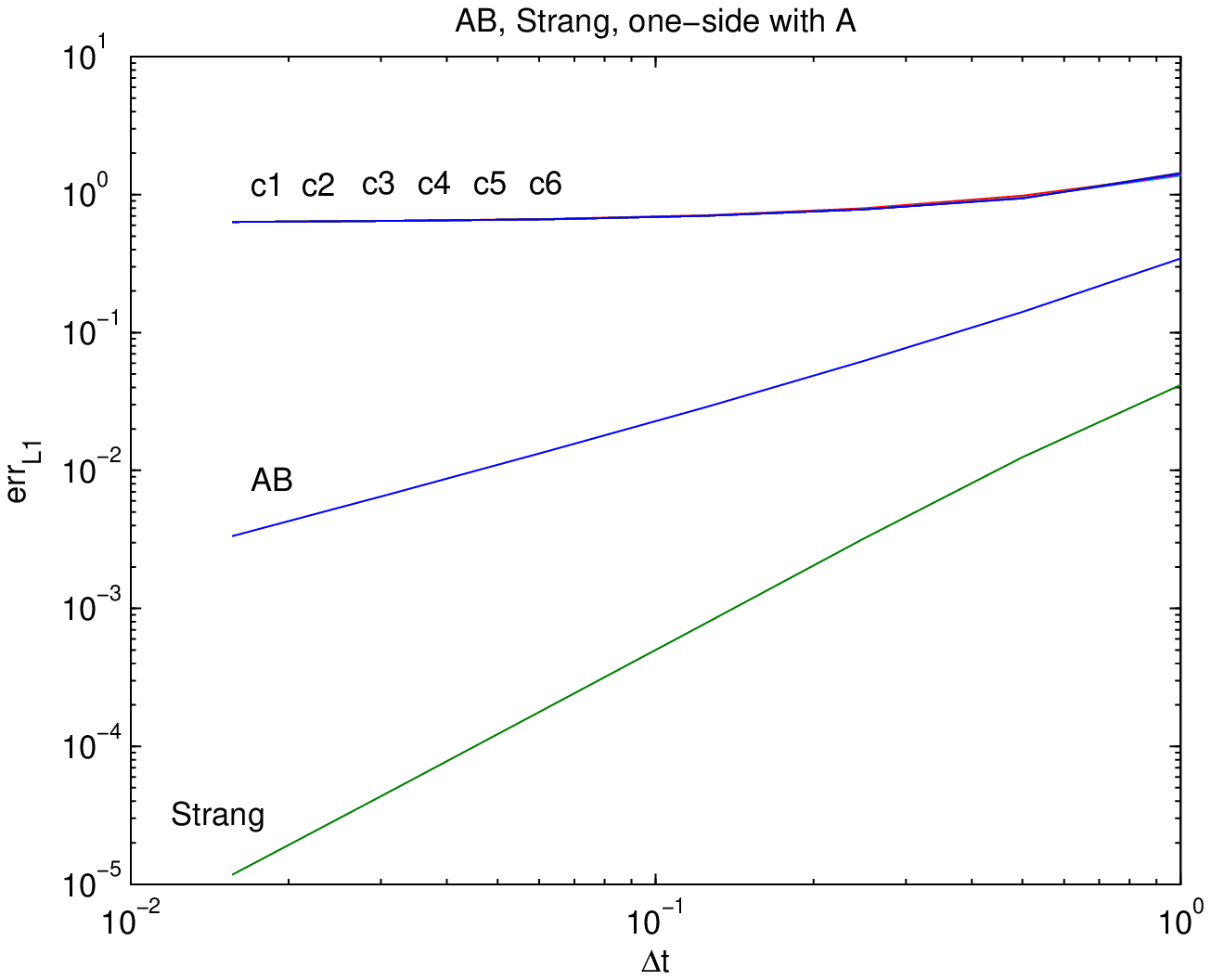} 
\includegraphics[width=9.0cm,angle=-0]{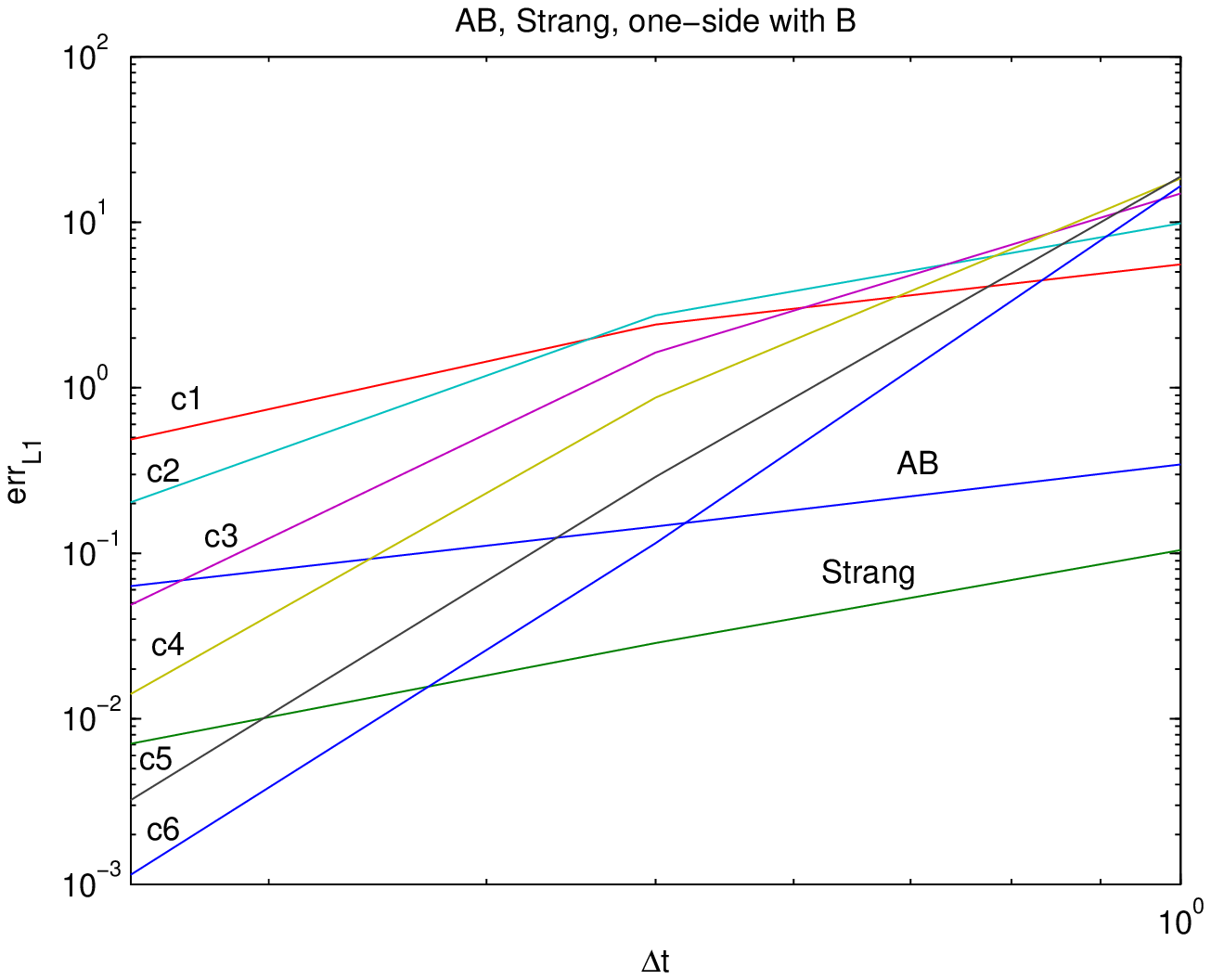} 
\includegraphics[width=9.0cm,angle=-0]{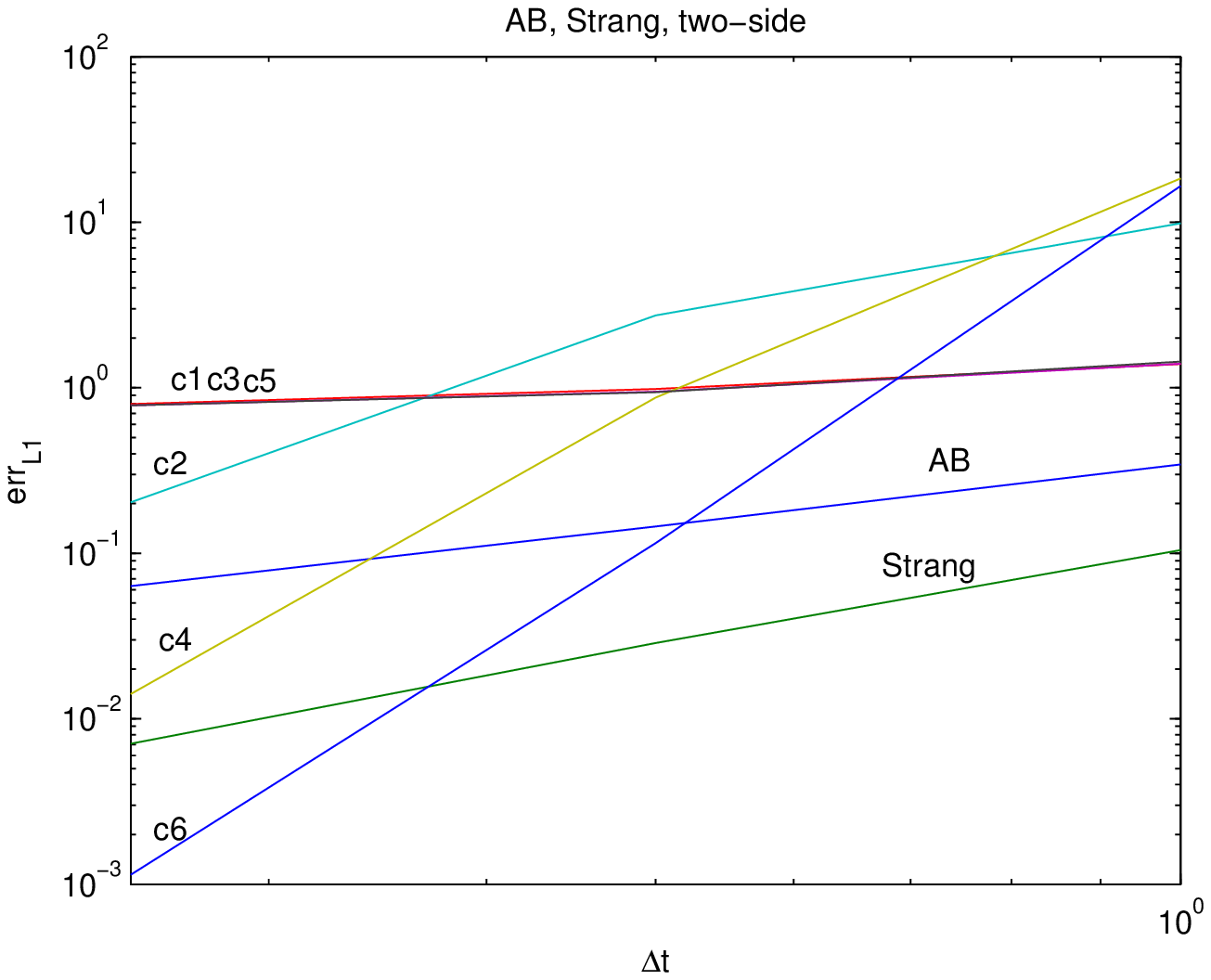} 
\end{center}
\caption{\label{two_phase} Numerical errors of the one-side Splitting scheme with $A$ (upper figure),  the one-side Splitting scheme with $B$ (middle figure) and the
iterative schemes with $1, \ldots, 6$ iterative steps (lower figure).}
\end{figure}

\begin{remark}
For all iterative schemes, we can reach faster results as for the
The iterative schemes with fast computations of the exponential matrices
standard schemes.
With $4-5$ iterative steps we obtain more accurate results as we did
for the expensive standard schemes.
With one-side iterative schemes we reach the best convergence results.
\end{remark}

In the following, we present a multi-layer model in the underlying rock
and assume multiple heat sources.
The aim is to see a distribution of the heat in the upper-lying earth-layers.

\subsection{Parameters of the model equations}

In the following all parameters of the model equations (\ref{mobile1})-(\ref{adsorpt1_immo})
are given in Table \ref{model_par_1}.
\begin{table}[h]
\begin{center}
\begin{tabular}{||c|c||}
\hline
density & $\rho = 1.0$ \\
mobile porosity & $\phi = 0.333$ \\
immobile porosity & $0.333$  \\
Diffusion & $D = 0.0$ \\
longitudinal Dispersion & $\alpha_L = 0.0$ \\
transversal Dispersion & $\alpha_T = 0.00$ \\
Retardation factor & $R = 10.0e-4$  (Henry rate). \\
Velocity field & $\vec{v} = ( 0.0,  4.0 \; 10^{-3})^t$. \\
Decay rate of the 1st heat source & $\lambda_{AB} = 1 \; 10^{-68}$. \\
Decay rate of the 2nd heat source & $\lambda_{AB} = 2 \; 10^{-3}$, $\lambda_{BNN} = 1 \; 10^{-68}$. \\
Decay rate of the 3rd heat source & $\lambda_{AB} = 0.25 \; 10^{-3}$, $\lambda_{CB} = 0.5 \; 10^{-3}$. \\
\hline
\hline
Geometry (2d domain) & $\Omega = [0, 100] \times [0, 100]$. \\
Boundary   & Neumann boundary at  \\
           & top, left and right boundaries. \\
           & Outflow boundary  \\
           & at the bottom boundary \\
\hline
\end{tabular}
\caption{\label{model_par_1} Model-Parameters.}
\end{center}
\end{table}

The discretization and solver method are given as:

For the spatial discretization method, we apply Finite volume methods
 of 2nd order, with the following parameters in Table \ref{model_par_1}.
\begin{table}[h]
\begin{center}
\begin{tabular}{||c|c||}
\hline
spatial step size & $\Delta x_{min} = 1.56,\Delta x_{max} = 2.21$ \\
refined levels & 6 \\
Limiter & Slope limiter \\
Test functions & linear test function \\
               & reconstructed with neighbor gradients \\
\hline
\end{tabular}
\caption{\label{model_par_1} Spatial discretization parameters.}
\end{center}
\end{table}

For the time discretization method, we apply 
Crank-Nicolson method (2nd order), with the following parameters in Table \ref{model_par_2}.
\begin{table}[h]
\begin{center}
\begin{tabular}{||c|c||}
\hline
Initial time-step & $\Delta t_{init} =  5 \; 10^2 $ \\
controlled time-step & $\Delta t_{max} =  1.298 \; 10^2,\Delta t_{min} =  1.158 \; 10^2  $ \\
Number of time-steps & $100,80,30,25$ \\
Time-step control & time steps are controlled with \\
                 & the Courant-Number $\mbox{CFL}_{max} = 1$ \\
\hline 
\end{tabular}
\caption{\label{model_par_2} Time discretization parameters.}
\end{center}
\end{table}

For the discretised equations are solved with the
following methods, see the description in Table \ref{model_par_3}.

\begin{table}[h]
\begin{center}
\begin{tabular}{||c|c||}
\hline
Solver & BiCGstab (Bi conjugate gradient method) \\
Preconditioner & geometric Multi-grid method \\
Smoother & Gauss-Seidel method as smoothers for \\
         & the Multi-grid method \\
Basic level & $0$ \\
Initial grid & Uniform grid with $2$ elements \\
Maximum Level &  $6$ \\
Finest grid & Uniform grid with $8192$ elements \\
\hline
\end{tabular}
\caption{\label{model_par_3} Solver methods and their parameters.}
\end{center}
\end{table}

For the numerical experiments, we discuss the heat flow
of different heat sources in the underlying multiple domain
regime.

The underlying software tool is $r3t$, which was developed to
solve discretised partial differential equations.
We use the tool to solve transport-reaction equations, see \cite{fein04}.

\subsection{Temperatur in an underlying Rock with permeable and less permeable layers}

In the following we discuss the simulation with a porous media 
given in Figure \ref{layer_0}.
The velocity is given in vertical direction, the area of the 
domain is $[0, 100] \times [0, 80]$. 
\begin{figure}[H]
\begin{center} 
\includegraphics[width=14.0cm,angle=-0]{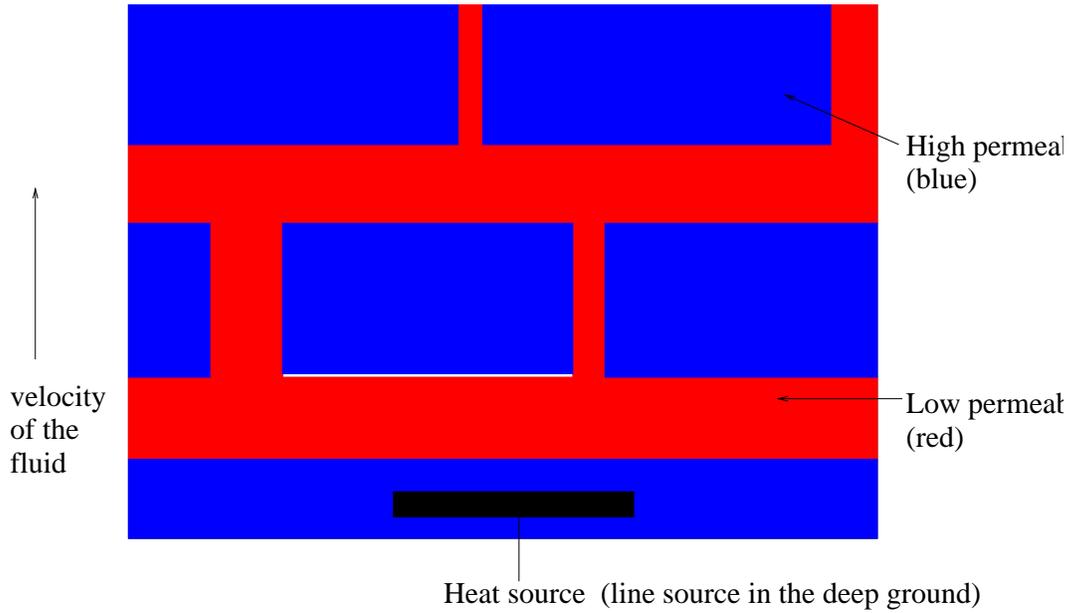}
\end{center}
\caption{\label{layer_0} Multiple layer regime of the underlying rocks and earth layers.}
\end{figure}

In the following Figure \ref{part_1} and \ref{part_2}, we present an example 
of the concentration of three inflow sources $x_{Source 1}, y_{Source 1} = (30, 75)$,
 $x_{Source 2}, y_{Source 2} = (50, 75)$ and  $x_{Source 3}, y_{Source 3} = (70, 75)$. 
The velocity is given perpendicular
in the underlying layers.
\begin{figure}[H]
\begin{center} 
\includegraphics[width=10.0cm,angle=-180]{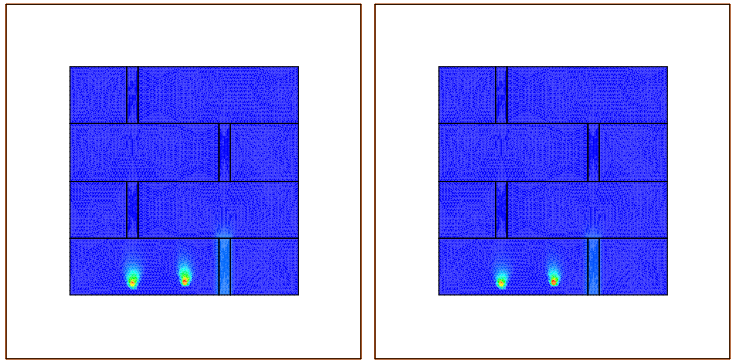}
\end{center}
\caption{\label{part_1} Three inflow sources $x_{Source 1}, y_{Source 1} = (30, 75)$,
 $x_{Source 2}, y_{Source 2} = (50, 75)$ and  $x_{Source 3}, y_{Source 3} = (70, 75)$ with perpendicular velocity and $2$ time-steps (initialization).}
\end{figure}
\begin{figure}[H]
\begin{center} 
\includegraphics[width=10.0cm,angle=-180]{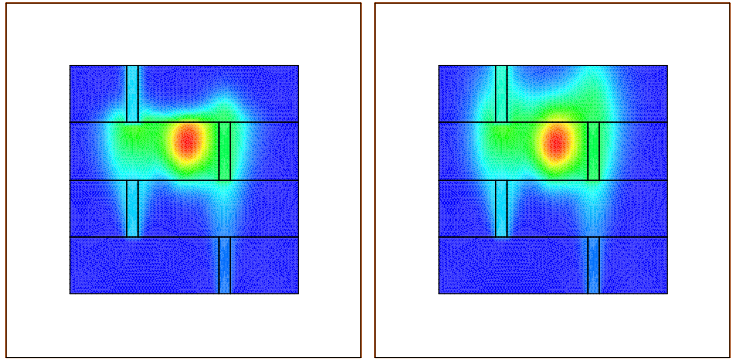} 
\end{center}
\caption{\label{part_2} Three inflow sources $x_{Source 1}, y_{Source 1} = (30, 75)$,
 $x_{Source 2}, y_{Source 2} = (50, 75)$ and  $x_{Source 3}, y_{Source 3} = (70, 75)$ with perpendicular velocity and $150$ time-steps (end phase).}
\end{figure}

\begin{remark}
The numerical experiments can also be fitted to real-life experiments. 
The problems are to achieve the correct diffusion and velocity-drift 
coefficients.
The fare field simulations, we obtain that the temperature derivations
are centered to the middle of the high permeable layers (in our case the
layers with high heat conduction).
Such prognostic results are important to allow an overview, how the
heat flow is distributed in the nearer earth-layers.
\end{remark}

\section{Conclusions and discussions}
\label{concl}

We have presented a continuous model for the multiple phases, we assumed that the
heat flow has a fluid behavior with exchange rates to adsorbed and 
immobile phases based on the different layers.

From the methodology side of the numerical simulations, the contributions were to decouple the
multiphase problem into single phase problems, where each single problem can be solved with more
accuracy.
The iterative schemes allows of coupling the simpler equations and for each additional iterative step, we could reduce the
splitting error.
Such iterative methods allow of accelerating the solver process of multiphase problems.

We can see in the numerical
experiments a loss of the heat transfer to impermeable layer and
strong temperature gradients within permeable layers.

\bibliographystyle{mychicago}
\bibliography{geiser_heat_2012.bib}

\begin{thebibliography}{}

\bibitem[\protect\citeauthoryear{Farago~I}{Farago~I}{2005}]{fargei05}
Farago~I Geiser~J (2005).
\newblock Iterative operator-splitting methods for linear problems.
\newblock Technical Report 1043, Weierstrass Institute for Applied Analysis and
  Stochastics, Berlin, Germany, Mohrenstrasse, Berlin, Germany.

\bibitem[\protect\citeauthoryear{Fein}{Fein}{2004}]{fein04}
Fein E (2004).
\newblock Software package $r^3t$.
\newblock Technical report.

\bibitem[\protect\citeauthoryear{Frolkovi\v{c}}{Frolkovi\v{c}}{2002a}]{fro02}
Frolkovi\v{c} P (2002a).
\newblock Flux-based method of characteristics for contaminant transport in
  flowing groundwater.
\newblock {\em Computing and Visualization in Science\/}~{\em 5\/}(2), pp.
  73--83.

\bibitem[\protect\citeauthoryear{Frolkovi\v{c}}{Frolkovi\v{c}}{2002b}]{fro02_b}
Frolkovi\v{c} P (2002b).
\newblock Flux-based methods of characteristics for transport problems in
  groundwater flows induced by sources and sinks.
\newblock {\em Computational Methods in Water Resources (S.M. Hassanizadeh et
  al.)Volume II., Elsevier, Amsterdam, Boston, Heidelberg,\/}~{\em 2}, pp.
  979--986.

\bibitem[\protect\citeauthoryear{Frolkovi\v{c} and De~Schepper}{Frolkovi\v{c}
  and De~Schepper}{2001}]{frodesch01}
Frolkovi\v{c} P and De~Schepper H (2001).
\newblock Numerical modelling of convection dominated transport coupled with
  density driven flow in porous media.
\newblock {\em Advances in Water Resources\/}~{\em 24}, pp. 63--72.

\bibitem[\protect\citeauthoryear{Frolkovi\v{c} and Geiser}{Frolkovi\v{c} and
  Geiser}{2003}]{fro_gei02}
Frolkovi\v{c} P and Geiser J (2003).
\newblock Discretization methods with discrete minimum and maximum property for
  convection dominated transport in porous media.
\newblock I.~Dimov, I.~Lirkov, S.~Margenov and Z.~Zlatev (eds.), Numerical
  Methods and Applications, 5th International Conference, NMA 2002, Borovets,
  Bulgaria. Berlin, Heidelberg, pp.\  446--453.

\bibitem[\protect\citeauthoryear{Geiser}{Geiser}{2003}]{gei_diss03}
Geiser J (2003).
\newblock {\em Gekoppelte Diskretisierungsverfahren f\"ur Systeme von
  Konvektions-Dispersions-Diffusions-Reaktionsgleichungen.}
\newblock Ph.\ D. thesis, Universit\"at Heidelberg.

\bibitem[\protect\citeauthoryear{Geiser}{Geiser}{2006}]{gei_06}
Geiser J (2006).
\newblock Discretisation methods with analytical solutions for
  convection-diffusion-dispersion-reaction-equations and applications.
\newblock {\em Journal of Engineering Mathematics.\/}~{\em 57\/}(1), pp.
  79--98.

\bibitem[\protect\citeauthoryear{Geiser}{Geiser}{2009}]{geiser_book_09}
Geiser J (2009).
\newblock {\em Decomposition Methods for Differential Equations: Theory and
  Applications}.
\newblock Chapman \& Hall/CRC Numerical Analysis and Scientific Computing
  Series, edited by Magoules and Lai.
\newblock First Edition.

\bibitem[\protect\citeauthoryear{Glowinski}{Glowinski}{2003}]{glow03}
Glowinski R (2003).
\newblock {\em Numerical methods for fluids.}
\newblock Handbook of Numerical Analysis, Gen. eds. P.G.~Ciarlet, J.~Lions,
  Vol.~IX, North-Holland Elsevier, Amsterdam, The Netherlands.

\bibitem[\protect\citeauthoryear{Gobbert and Ringhofer}{Gobbert and
  Ringhofer}{1998}]{gobb96}
Gobbert MK and Ringhofer CA (1998).
\newblock An asymptotic analysis for a model of chemical vapor deposition on a
  microstructured surface.
\newblock {\em SIAM Journal on Applied Mathematics.\/}~{\em 58}, pp. 737--752.

\bibitem[\protect\citeauthoryear{Hairer and Wanner}{Hairer and
  Wanner}{1996}]{hai96}
Hairer E and Wanner G (1996).
\newblock {\em Solving Ordinary Differential Equatons II}.
\newblock SCM, Springer-Verlag Berlin-Heidelberg-New York.
\newblock Second Edition.

\bibitem[\protect\citeauthoryear{Hairer and Wanner}{Hairer and
  Wanner}{1992}]{hai92}
Hairer, E Norsett~SP and Wanner G (1992).
\newblock {\em Solving Ordinary Differential Equatons I}.
\newblock SCM, Springer-Verlag Berlin-Heidelberg-New York.
\newblock Second Edition.

\bibitem[\protect\citeauthoryear{Hansen and Ostermann}{Hansen and
  Ostermann}{2009}]{han08}
Hansen E and Ostermann A (2009).
\newblock Exponential splitting for unbounded operators.
\newblock {\em Mathematics of Computation\/}~{\em 78}.

\bibitem[\protect\citeauthoryear{Jahnke and Lubich}{Jahnke and
  Lubich}{2009}]{jan00}
Jahnke T and Lubich C (2009).
\newblock Error bounds for exponential operator splittings.
\newblock {\em BIT Numerical Mathematics\/}~{\em 40\/}(4), pp. 735--745.

\bibitem[\protect\citeauthoryear{Kanney and Kelley}{Kanney and
  Kelley}{2003}]{kan03}
Kanney, J Miller~C and Kelley CT (2003).
\newblock Convergence of iterative split-operator approaches for approximating
  nonlinear reactive transport problems.
\newblock {\em AAdvances in Water Resources\/}~{\em 26}, pp. 247--261.

\bibitem[\protect\citeauthoryear{Rouch}{Rouch}{2006}]{rouch06}
Rouch H (2006).
\newblock Mocvd research reactor simulation.
\newblock {\em Proceedings of the COMSOL Users Conference 2006 Paris}, Paris,
  France.

\bibitem[\protect\citeauthoryear{Vabishchevich}{Vabishchevich}{2011}]{vab2011}
Vabishchevich PN (2011).
\newblock A new class of additive (splitting) operator-difference schemes.
\newblock {\em Mathematics of Computations\/}~{\em 81\/}(277), pp. 267--276.

\end{thebibliography}

\end{document}